\documentclass{article}

\title{Tight Bound for the Number of Distinct Palindromes in a Tree%
\footnote{This is a full version of a paper presented at SPIRE 2015~\cite{DBLP:conf/spire/GawrychowskiKRW15}.}}

\author{Pawe\l{} Gawrychowski\\
\small Institute of Computer Science\\[-0.8ex]
\small University of Wroc\l{}aw\\[-0.8ex] 
\small Poland\\
\small\tt gawry@cs.uni.wroc.pl\\
\and
Tomasz Kociumaka\\
\small Department of Computer Science\\[-0.8ex]
\small Bar-Ilan University\\[-0.8ex]
\small Ramat Gan, Israel\\
\small\tt kociumaka@mimuw.edu.pl\\
\and 
Wojciech Rytter \qquad Tomasz~Wale\'n\\
\small Institute of Informatics\\[-0.8ex]
\small University of Warsaw\\[-0.8ex]
\small Poland\\
\small\tt \{rytter,walen\}@mimuw.edu.pl}

\usepackage[bookmarks=false]{hyperref}
\usepackage{amsmath,amsfonts,amsthm}
\usepackage[T1]{fontenc}
\usepackage[utf8]{inputenc}
\usepackage{verbatim}
\usepackage{latexsym}
\usepackage{graphicx}
\usepackage{tikz}
\usetikzlibrary{arrows}
\usepackage{todonotes}
\usepackage[ruled]{algorithm2e}
\usepackage[capitalize]{cleveref}

\newcommand{\pal}{\textsf{pal}}

\newcommand{\val}{\mathrm{val}}
\newcommand{\dist}{\mathrm{dist}}

\newcommand{\DD}{\mathcal{D}}
\newcommand{\Oh}{\mathcal{O}}

\newcommand{\per}{\mathrm{per}}
\newcommand{\PalindromeTesting}{\textsc{PalindromeTest}}
\newcommand{\FindLongestPalindrome}{\textsc{FindLongest}}
\newcommand{\ReportAllPalindromes}{\textsc{ReportAll}}
\newcommand{\ReportAllEvenPalindromes}{\textsc{ReportAllEven}}
\newcommand{\perLen}{\mathsf{perLen}}
\newcommand{\distFunc}{\mathsf{dist}}
\newcommand{\upFunc}{\mathsf{up}}
\newcommand{\centerFunc}{\mathsf{center}}
\newcommand{\isAncestor}{\mathsf{isAncestor}}
\newcommand{\isPalindrome}{\mathsf{isPalindrome}}
\newcommand{\isEqual}{\mathsf{isEqual}}
\newcommand{\labelFunc}{\mathsf{label}}
\newcommand{\existsFunc}{\mathsf{exists}}
\newcommand{\childFunc}{\mathsf{child}}

\newtheorem{theorem}{Theorem}[section]
\newtheorem{lemma}[theorem]{Lemma}
\newtheorem{observation}[theorem]{Observation}

\newtheorem{claim}[theorem]{Claim}
\newtheorem{fact}[theorem]{Fact}

\theoremstyle{remark}

\newtheorem{problem}[theorem]{Problem}
\newcommand{\D}{\Psi}

\begin{document}

\maketitle

\begin{abstract}
    For an undirected tree with $n$ edges labelled by single letters, we consider its substrings,
    which are labels of the simple paths between pairs of nodes. 
    A palindrome is a word $w$ such that $w=w^R$, where $w^R$ denotes the reverse of $w$.
    We prove that $PAL(n)=\Oh(n^{1.5})$, where $PAL(n)$ denotes the number of  distinct 
palindromic substrings
in a tree of size $n$.
    This solves an open problem of Brlek, Lafrenière, and Provençal  (DLT 2015
\cite{Brlek}), who showed that $PAL(n)=\Omega(n^{1.5})$.
    Hence, we settle the tight bound of $\Theta(n^{1.5})$ for the maximum palindromic complexity of
    trees. For standard strings, i.e., for itrees which are simple paths, the palindromic complexity is exactly $n+1$.
    
    We also propose $\Oh(n^{1.5} \log{n})$-time algorithm for reporting all distinct palindromes 
     and $O(n\, polylog n)$ time algorithm for palindrom testing and finding the
longest palindrom in a tree.
    \end{abstract}
    
    \section{Introduction}
    Regularities in words are extensively studied in combinatorics and text algorithms.
    One of the basic types of such structures are palindromes: symmetric words, the ones 
which are the same when read in both directions. 
    The \emph{palindromic complexity} of a word is the number of distinct palindromic substrings in the word.
    An elegant argument shows that the palindromic complexity of a word of length $n$ does not exceed $n+1$~\cite{DBLP:journals/tcs/DroubayJP01},
    which is already attained by a unary word $\texttt{a}^n$.
    Therefore the problem of palindromic complexity for words is completely settled, and a natural next
    step is to generalize it to trees.
    
    In this paper, we consider the palindromic complexity of undirected trees with edges labelled by single letters.
    We define substrings of such a tree as the labels of simple paths between arbitrary two nodes.
    Each label is the concatenation of the labels of all edges on the path.
Denote by $pals(T)$ the set of all palindromic substrings of a tree $T$ and by 
$PAL(n)$ the maximum value of $pals(T)$ over all trees with $n$ edges.

    Fig.~\ref{fig:example} illustrates palindromic substrings in a sample tree.
    Note that palindromes in a word of length $n$ naturally correspond to palindromic substrings in a path of $n$ edges.

    \begin{figure}[ht]
        \begin{center}
        \begin{tikzpicture}
    \input{_fig_example.tex}

    \tikzstyle{stgreen} = [draw, circle, fill=black, minimum size = 6pt, inner sep = 0 pt, color=green!50!black]
    \tikzstyle{lgreen} = [draw, ultra thick, color=green!50!black]

    \node[stred] at (n7) {};
    \node[stred] at (n4) {};
    \node[stred] at (n2) {};
    \draw[lred] (n7)--(n6)--(n4)--(n2);

    \node[stgreen] at (n1) {};
    \node[stgreen] at (n3) {};
    \node[stgreen] at (n8) {};
    \node[stgreen] at (n12) {};
    \node[stgreen] at (n13) {};
    \draw[lgreen] (n1)--(n3)--(n6)--(n8)--(n12)--(n13);

\end{tikzpicture}
        \end{center}
        \caption{
            A sample tree $T$. We have $pals(T)=\{$
            {$a$}, {$b$}, {$c$},
            {$aa$},
            {$aca$},
            {$acaaca$},
            {$bcb$},
            {$bccb$},
            {$caac$},
            {$cbc$},
            {$cbcbc$},
            {$cc$}
            $\}$.
            An occurrence of a palindrome $aca$ is marked red, and an occurrence of a palindrome $cbcbc$ is marked green.
        }
        \label{fig:example}
    \end{figure}

    The study of the palindromic complexity of trees was recently initiated by Brlek, Lafrenière, and Provençal~\cite{Brlek}, who
    constructed a family of trees with $n$ edges containing $\Theta(n^{1.5})$ distinct palindromic substrings.
    They conjectured that there are no trees with asymptotically larger palindromic complexity and proved this claim 
    for a special subclass of trees.
    
    \paragraph{Our Result}
    We show that $PAL(n)= \Oh(n^{1.5})$. This bound is tight by the construction given in~\cite{Brlek};
    hence, we completely settle the asymptotic maximum palindromic complexity for trees.
    We also provide $\Oh(n^{1.5}\log n)$ algorithm for reporting all distinct palindromes
and $O(n\, polylog n)$ time algorithm for palindrom testing and finding the
longest palindrom in a tree.
    
    \paragraph{Related Work}
    Palindromic complexity of words was studied in various aspects. This includes
    algorithms determining the complexity~\cite{DBLP:journals/ipl/GroultPR10},
    bounds on the average complexity~\cite{DBLP:journals/dm/AnisiuAK10}, and generalizations to circular words~\cite{DBLP:journals/tcs/Simpson14}.
    Finite and infinite palindrome-rich words received particularly high attention; 

    see e.g.~\cite{DBLP:journals/ijfcs/BrlekHNR04,DBLP:journals/tcs/DroubayJP01,DBLP:journals/ejc/GlenJWZ09}.
    This class contains, for example, all episturmian and thus all Sturmian words~\cite{DBLP:journals/tcs/DroubayJP01}.
    
    Recently, some almost exact bounds for the number of distinct palindromes in star-like trees have been shown 
    by Glen et al.~\cite{GlenSS19}.
    Also the palindromes in directed trees have been studied by Funakoshi et al.
    \cite{DBLP:conf/stringology/FunakoshiNIBT19} who presented $\Oh(n\log h)$ time algorithm
    to compute all maximal palindromes and all distinct palindromes in a TRIE $T$ of height $h$.

    In the setting of labelled trees, other kinds of regularities were also studied. It has been shown that a tree with $n$ edges contains $\Oh(n^{4/3})$ distinct squares~\cite{DBLP:conf/cpm/CrochemoreIKKRRTW12} and $\Oh(n)$ distinct cubes~\cite{DBLP:journals/algorithmica/KociumakaRRW17}.
    Both bounds are known to be tight.
    Interestingly, the lower bound construction for squares resembles that for palindromes~\cite{Brlek}.

     \paragraph{Outline of the Paper}
     In \Cref{sec:preliminaries}, we introduce basic terminology and combinatorial toolbox.
     Next, in \Cref{sec:lower-bounds}, we quickly summarize previously known results
     for the lower bounds on the number of distinct palindromes.

     In \Cref{sec:spine-trees}, we introduce the special family of the trees
     called {\em spine trees} and prove th upper bounds for those trees.
     In \Cref{sec:double-trees}, we show how every tree can be decomposed into
     spine trees, and in the \Cref{sec:main}, we combine those results to obtain
     the upper bound on the number of distinct palindromes.

     In \Cref{sec:alg}, we introduce algorithmic toolbox and provide an algorithm for
     reporting all distinct palindromes.
    \section{Preliminaries}
    \label{sec:preliminaries}
    
    
    A word $w$ is a sequence of characters $w[1],w[2],\ldots,w[|w|]\in\Sigma$, often denoted $w[1..|w|]$.
    A substring of $w$ is any word of the form $w[i..j]$, and if $i=1$ ($j=|w|$), then it is called a prefix (a suffix, respectively).
    A period of $w$ is an integer $p$, $1\le p \le |w|$, such that $w[i]=w[i+p]$ for $i=1,2,\ldots,|w|-p$.
    The shortest period of $w$, denoted $\per(w)$, is the smallest such $p$.
\subsection{Some combinatorics of words}
    The following well known periodicity lemma, characterizes the properties of periods.

    \begin{lemma}[Periodicity Lemma~\cite{fine1965uniqueness}]
        \label{lem:periodicity}
        If $p$, $q$ are periods of a word $w$ of length $|w| \ge p + q - gcd(p, q)$,
        then $gcd(p, q)$ is also a period of $w$.
    \end{lemma}
    
    The following lemma is a straightforward consequence of the Periodicity Lemma (\Cref{lem:periodicity}).

    \begin{lemma}\label{lem:periodicperiod}
    Suppose a word $v$ is a substring of a longer word $u$ 
which has a period $p\le \frac{1}{2}|v|$.
    Then $\per(u)=\per(v)$.
    \end{lemma}
    \begin{proof}
    Let us assume that $p_u=\per(u)$, $p_v=\per(v)$ and $p_u\not=p_v$.
    Since $v$ is a substring of $u$, and $p$ is a period of both $u$, 
    we have clearly that $$p_v \le p_u \le p \le \frac{1}{2}|v|.$$
    The word $v$ and periods $p_v$ and $p_u$ met the conditions of the Periodicity Lemma,
    so the $gcd(p_v, p_u)$ is also a period of $v$.

    Since $p_v$ is the minimal period, the $p_u=a\cdot p_v$ for some $a>1$.
    But in such case the $p_v$ is also a period of whole word $u$ --- contradiction.
    \end{proof}
    
    We have the following connection between periods and palindromes.
    \begin{observation}
    \label{obs:palperiod}
    Suppose a palindrome $v$ is a suffix of a longer palindrome $u$.
    Then $v$ is a prefix of $u$ and thus $|u|-|v|$ is a period of $u$ and of $v$.
    \end{observation}
    
   \subsection{Centroid decomposition}
For a tree $T$ and its node $r$ denote by $pals(T,r)$ the set of palindromic substrings of $T$
corresponding to simple paths containing the node $r$. 

\begin{lemma}\label{centroids}
If $|pals(T,r)|=O(|T|^{1+\epsilon})$, for $\epsilon>0$, then $PAL(n)=O(n^{1+\epsilon}$.
\end{lemma}
\begin{proof}
We follow the approach from~\cite{DBLP:conf/cpm/CrochemoreIKKRRTW12}.
We use the folklore fact that 
every tree $T$ on $n$ edges contains a \emph{centroid} node $r$ such that every component of $T\setminus\{r\}$
is of size at most $\frac{n}{2}$.
We separately count palindromic substrings corresponding to the paths going through the centroid
$r$ and paths fully contained in a single component of $T\setminus\{r\}$. 
Finally, we obtain the following recurrence for $\pal(n)$, the maximum number of palindromes in a tree with $n$ edges:
\[
PAL(n) = \Oh(n^{1+\epsilon})+\max\left\{\sum_i PAL(n_i) : \forall_i\, n_i \leq \tfrac{n}{2}\: \text{and} \:  \sum_i n_i <n \right\}.
\]
It solves to $PAL(n)=\Oh(n^{1+\epsilon})$.
\end{proof}
    
    \subsection{D-Trees}
   For a tree $T$ and its node $r$ denote by $pals(T,r)$ the set of palindromic substrings of $T$
corresponding to simple paths containing the node $r$. 
We consider directed acyclic graphs, named D-trees,  such that $pals(T,r)$ is a subset of
palindromic strings corresponding to simple directed paths  in such graphs containg 
the node $r$.
    
    Define a {\em double tree}  $\DD=(T_{\ell},T_r,r)$ as a labelled tree consisting of
    two trees $T_{\ell}$ and $T_r$ sharing a common root $r$ but otherwise disjoint.
    The edges of $T_{\ell}$ and $T_r$ are directed to and from $r$, respectively.
    The size of $\DD$ is defined as $|\DD| = |T_{\ell}|+|T_r|$.

    For any $u,v\in\DD$, we denote by $path(u,v)$ the path drom $u$ to $v$ 
and by $\val(u,v)$ denote the sequence of the labels of edges on this path.
    We say that a path is palindromic if it corresponds to a palindromic word.
Denote by $pals(D)$ the set of palindromic substrings of a D-tree $D$.

A \emph{substring} of $\DD$ is any word $\val(u,v)$ such that $u\in T_{\ell}$ and $v\in T_r$.
    Let $$\dist(u,v)=|\val(u,v)| \ \mbox{and}\ \per(u,v)=\per(\val(u,v))$$
    
    We consider only \emph{deterministic} double trees ({\em D-trees}, in short), meaning that
    all the edges outgoing from a node have distinct labels, and similarly
    all the edges incoming into a node have distinct labels.
    An example of such a double tree is shown in Fig.~\ref{fig:combined}.

    \begin{figure}[ht]
    \begin{center}
    \includegraphics[width=\textwidth]{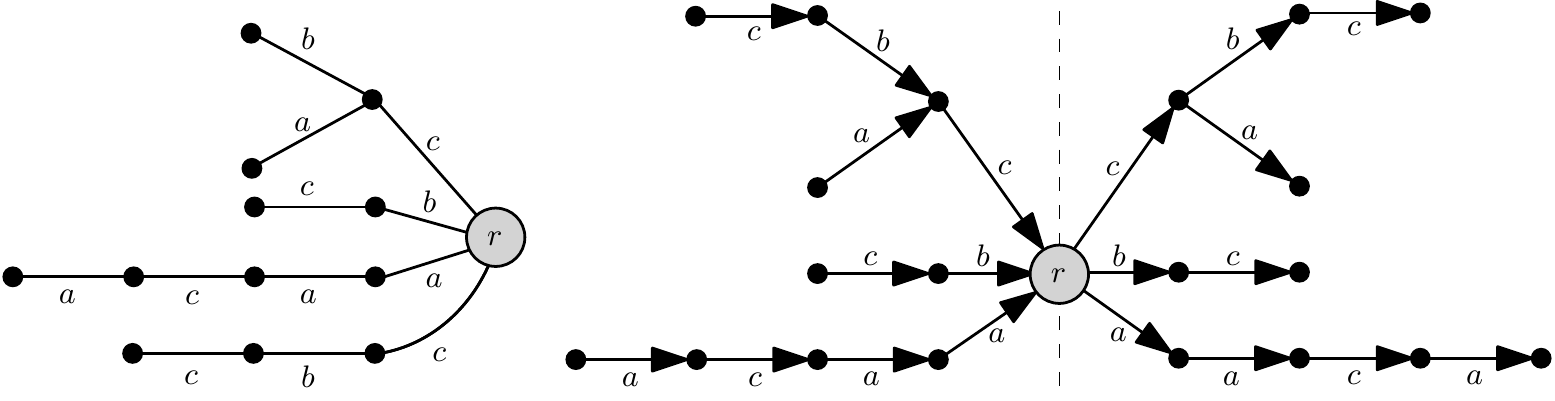}
    \end{center}
    \caption{To the left: an example undirected tree with $pals(T,r)$ 
containing $9$ palindromic substrings of length 2 or more
    $bcb,\; bccb,\; aca,\; cbc,\; caac,\; cc,\; cbcbc,\; aa,\; acaaca$.
    To the right: D-tree $\D(T,r)$ obtained after rooting the tree at $r$, merging
    both subtrees connected to $r$ with edges labelled by $c$, and duplicating the resulting tree.}
    \label{fig:combined}
    \end{figure}
    
\paragraph{\bf Trees $\Rightarrow$ D-Trees.}
For a tree $T$ and its node $r$ we construct a D-tree $\D(T,r)=(T_l,T_r,r)$ in the following way.

\medskip
We root $T$ at $r$ directing all the edges so that they point towards the root and then determinize
the resulting tree by gluing together two children of the same node whenever their edges have the same
label. Finally, we create a D-tree by duplicating the tree and changing the
directions of the edges in the second copy; see \cref{fig:combined} for a sample application of this process. 

It is easy to see that for any simple path from $u$ to $v$
going through $r$ in the original tree we can find $u'\in T_{\ell}$ and $v'\in T_r$ such that
$\val(u,v)=\val(u',v')$. It implies the following fact.
\begin{fact}\label{Psi} $pals(T,r)\subseteq pals(\D(T,r))$.
\end{fact}

\section{Proof of $O(n^{1.5})$ Upper Bound }
Due to Lemma~\ref{centroids} it is enough to consider only palindromic paths passing through
a fixed node $r$ of the tree. Then, due to Fact~\ref{Psi} this is reduced to the 
estimation of palindroms in a D-tree, which is easier.

\smallskip
We consider a D-tree $(T_l,T_r,r)$.
A directed palindromic subpath $path(u',v')$ of $path(u,v)$ is called its {\it central part}
iff $dist(u,u')=dist(v',v)$ and $u'=r$ or $v'=r$.
The end-nodes of the central part are called {\it paired} nodes.

The crucial role in our proof of the upper bound play D-trees 
called spine-trees.
A \emph{spine-tree} is a D-tree with a distinguished path, called \emph{spine},
joining vertices $s_\ell\in T_{\ell}$ and $s_r\in T_r$. Additionally, we insist
that this path cannot be extended preserving the period $p=\per(s_\ell,s_r)$.

\begin{figure}[ht]
\begin{center}
\includegraphics[width=\textwidth]{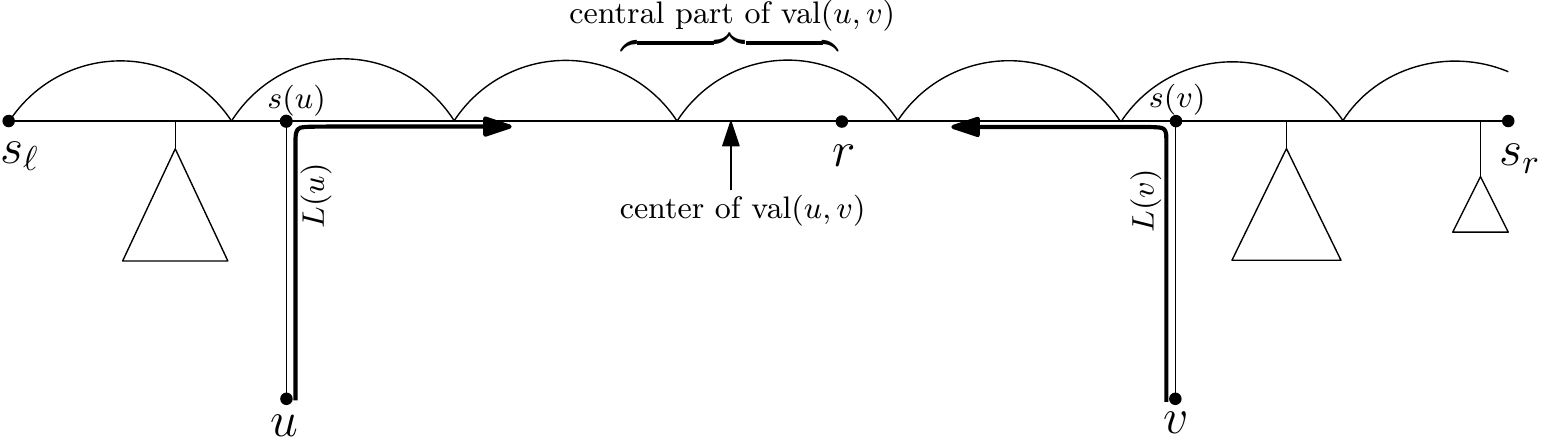}
\end{center}
\caption{A spine-tree, whose spine is the path from $s_\ell$ to $s_r$,
 with an induced palindrome $\val(u,v)$.
Observe that $L(u)=L(v)$ is a prefix of the palindrome. Note that $d(s(u),r)\ge p$ but $d(r,s(v))$ might be smaller than $p$.}
\label{fig:spine-tree}
\end{figure}

\medskip
By symmetry of the counting problem (up to edge reversal in a double tree), we consider later only
palindromic paths $path(u,v)$ such that
    $\dist(u,r)\ge \dist(r,v)$. The right end-point of the central part of each such path is the root.

\subsection{Combinatorial outline}
A palindromic substring is \emph{induced} by a spine-tree if its central part is a fragment of the spine
of length at least $p$, where $p$ is the period of the spine; 
see Fig.~\ref{fig:spine-tree} for an example.

\medskip The structure of the proof is described informally as follows.
\begin{itemize}
\item We show that the number of palindromes induced by a spine-tree
is $O(n^{1.5})$.
\item
We partition set of palindromes into so called {\it middle} palindromes and others. 
 The number of latter ones is easily estimated to be small.
\item 
The D-tree is partitioned into smaller D-trees, each with a distinguished spine.
The total size of all these D-subtrees is linear. 
\item 
Then we show that 
the set of all middle palindromes in a D-tree is a subset of the union of
palindromes induced by smaller spine D-subtrees.
\item 
Now the upper bound of all palindromes in a D-tree follows from
the upper bound on induced palindromes.
\end{itemize}

\subsection{Number of palindromes induced by a spine tree}
\label{sec:spine-trees}
For a node $u$ of the spine-tree, let $s(u)$ denote the nearest node of the spine
(if $u$ is already on the spine, then $u=s(u)$). Since the spine-tree is deterministic,
it satisfies the following property.
\begin{fact}\label{fct:deter}
For any induced palindrome $\val(u,v)$, the path $\val(s(u),s(v))$ is an inclusion--maximal
fragment of $\val(u,v)$ admitting period $p$. 
\end{fact}

\begin{lemma}
There are up to $n\sqrt{n}$ distinct palindromic substrings induced by a given spine-tree of size $n$.
\label{lem:spine-trees}
\end{lemma}

\begin{proof}
Define the \emph{label} $L(u)$ for a node $u\in T_{\ell}$ as the prefix of $\val(u,s_r)$ of length $\dist(u,s(u))+p$.
Similarly, the label $L(v)$ of a node $v\in T_r$ is the reversed suffix of $\val(s_\ell,v)$ of length $p+\dist(s(v),v)$.
We leave the label undefined if $\val(u,s_r)$ or $\val(s_\ell,v)$ is not sufficiently long,
i.e., if $d(s(u),s_r)<p$ or $d(s_\ell,s(v))<p$.

\smallskip
Consider a palindrome $\val(u,v)$ induced by the spine-tree. Fact~\ref{fct:deter} implies that 
that the fragment $\val(s(u), s(v))$
is a maximal fragment of $\val(u,v)$ with period $p$. Since the central part of the palindrome is of length at least
$p$ and lies within this fragment, the fragment must be symmetric, i.e., we must have $\dist(u,s(u))=\dist(s(v),v)$, and the labels of $u$ and $v$ are both defined. 

Consequently, $|L(u)|=|L(v)|$ and actually the labels $L(u)$ and $L(v)$ are equal.
Hence, to bound the number of distinct palindromes, we group together nodes with the same
labels. Let $V_L$ be the set of vertices of $T_\ell\cup T_r$ with label $L$.
We have the following claim.

\begin{claim}
For any label $L$,
there are at most $\min(|V_L|^2,n)$ distinct induced palindromes with endpoints in $V_L$.
\end{claim}
\begin{proof} (of the claim)\\
Consider all distinct induced palindromes $\val(u,v)$ such that $L(u)=L(v)=L$.
A substring is uniquely determined by the endpoints of its occurrence, so $|V_L|^2$ is an upper bound
on the number of these palindromes.

We claim that every such palindrome is also uniquely determined by its length, which immediately gives
the upper bound of $n$.
Indeed, $\dist(u,s(u))=\dist(s(v),v)=|L|-p$ and $\val(s(u),s(v))$ has period $p$,
so if the length is known, then $\val(s(u),s(v))$ can be recovered from its prefix of length~$p$, i.e., the suffix of $L$ of length $p$.
\end{proof}

The sets $V_L$ are disjoint, so by the above claim and using the inequality $\min(x,y)\le \sqrt{xy}$,
the number of distinct palindromes induced by the spine-tree is at most:
\[\sum_L \min(|V_L|^2,n)\le\sum_L \sqrt{|V_L|^2\cdot n} \le \sqrt{n}\cdot \sum_L |V_L| \le n^{1.5}.\qedhere\]
\end{proof}

\subsection{Number of all palindromes}
\label{sec:double-trees}

Consider a node $u\in T_{\ell}$ and all distinct palindromes $P_1,\ldots,P_k$ with an occurrence starting at $u$. 
Observe that their central parts $C_1,\ldots,C_k$ have distinct lengths:
indeed, $|P_i|=2\dist(u,r)-|C_i|$ and $\dist(u,r)\ge \frac12|P_i|$, so $\val(u,r)$ and $|C_i|$ determines the whole palindrome $P_i$.
Hence, we can order these palindromes so that $|C_1|>\ldots > |C_k|$, (i.e., $|P_1|<\ldots < |P_k|$).

Denote $\alpha=2\sqrt{n}$.
Palindromes $P_{2\alpha+1},\ldots,P_{k-\alpha}$ are called \emph{middle palindromes}.
There are $\Oh(\sqrt{n})$ remaining palindromes for fixed $u$ and $\Oh(n^{1.5})$ in total,
so we can focus on counting middle palindromes.  We start with the following characterization.

\begin{lemma}
Consider middle palindromes $P_i$ starting at a given node~$u$.
Central parts of these palindromes satisfy $|C_i|\ge \alpha$ and $\per(C_i)\le \frac12\sqrt{n}$.
Moreover, for each $P_i$ extending the central part $C_i$ by $\alpha$ characters in each direction preserves the shortest period.
\label{lem:characterization}
\end{lemma}

\begin{proof}
Since we excluded the $\alpha$ palindromes with the shortest central parts, the middle palindromes clearly have 
central parts of length at least $\alpha$.

\medskip \noindent
Let us now prove that $\per(C_{\alpha})\le \tfrac12\sqrt{n}$.

\smallskip \noindent
By Observation~\ref{obs:palperiod}, $|C_{j}|-|C_{j+1}|$ is a period of $C_{j}$ for $1\le j \le \alpha$.

\noindent
Since 
	$$\sum_{j=1}^{\alpha}(|C_{j}|-|C_{j+1}|) < |C_1|\le n,$$
for some $j$ we have 
	$$\per(C_j)\le |C_{j}|-|C_{j+1}|\le\tfrac12\sqrt{n}.$$
Moreover, $C_{\alpha}$ is a suffix of $C_j$, so the claim follows.

\medskip
For $i> 2\alpha$ (in particular, if $P_i$ is a middle palindrome), $C_i$ is a suffix of $C_{\alpha}$.
Additionally, for $2\alpha < i \le k-\alpha$ we can observe that $\per(C_{\alpha}) \le \frac{1}{2}{|C_i|}$
since $\per(C_{\alpha}) \le \frac12\sqrt{n}$ and $|C_i| \ge \alpha$.

Hence, we can apply Lemma~\ref{lem:periodicperiod} that implies $\per(C_i)=\per(C_{\alpha})$. 

\smallskip
\noindent
Moreover, 
$$|C_i|\le |C_{\alpha}|+\alpha-i< |C_{\alpha}|-\alpha,$$
so extending $C_i$ by $\alpha$ characters to the left preserves the period. 
By symmetry of $P_i$, the extension to the right also preserves the period. 
\end{proof}

\begin{figure}[ht]
\begin{center}
\includegraphics[width=\textwidth]{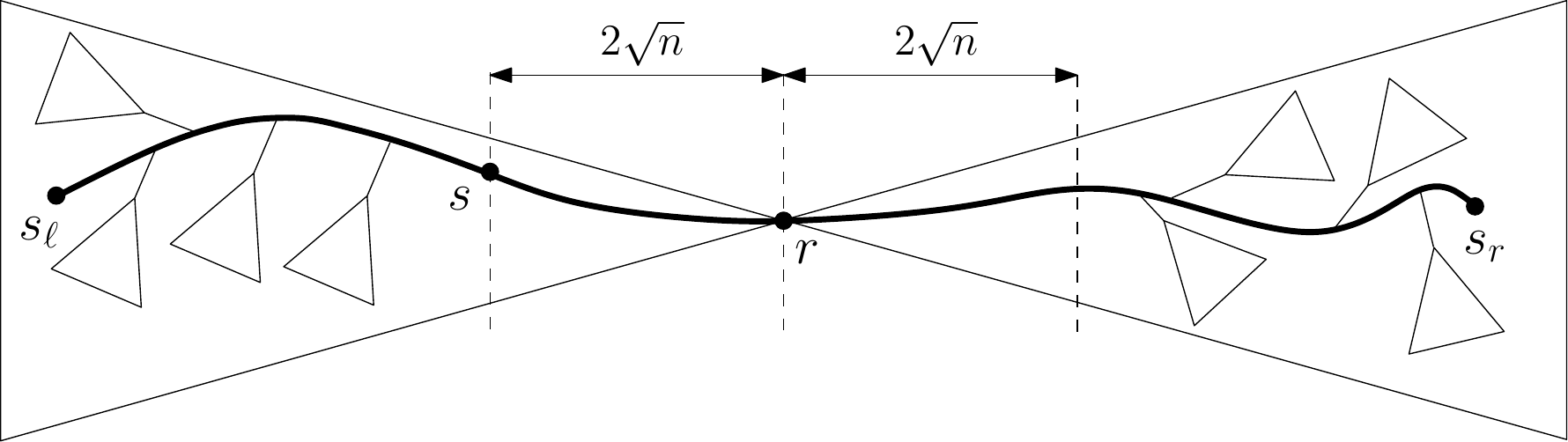}
\end{center}
\caption{A spine-tree constructed for a vertex $s$ in a D-tree.
Note that we do not attach subtrees at distance less than $\alpha$ from the root.}
\label{fig:decompose}
\end{figure}

Let us choose any $s\in T_{\ell}$ such that $$\dist(s,r)=\alpha \ \mbox{and}\ 
\per(s,r)\le \frac12\sqrt{n}.$$
Then, extend the period of $\val(s,r)$ to the left and to the right as far as possible, arriving at nodes $s_\ell$ and $s_r$, respectively. 

\medskip
We create a spine-tree with spine corresponding to the path
from $s_\ell$ to $s_r$ as shown in \cref{fig:decompose}. 
We attach to the spine all subtrees
hanging off the original path at distance at least $\alpha$ from the root.  In other words,
a vertex $u\in T_{\ell}$ which does not belong the spine is added to the spine-tree if $\dist(s(u),r)\ge \alpha$
and a vertex $v\in T_r$ --- if $\dist(r,s(v))\ge \alpha$. If $\dist(r,s_r)<\alpha$, then this procedure leaves no subtrees hanging in $T_r$
so we do not create any spine-tree for $s$. 

\smallskip
Now, let us consider a middle palindrome. By \cref{lem:characterization}, its central part satisfies $|C|\ge \alpha$ and 
$\per(C)\le \frac12\sqrt{n}$. Moreover, by \cref{lem:periodicperiod}, we have $\per(C)=\per(s,r)$ for the unique node $s\in T_{\ell}$ located within $C$
at distance $\alpha$ from the root. 

Consequently, $C$ lies on the spine of the spine-tree created for $s$, and $u$ belongs to a subtree attached to the spine. Additionally, since $C$ can be extended by $\alpha$ characters in each direction preserving the period, the other endpoint $v$ must also belong to such a subtree in $T_r$ (that is, we have $\dist(r,s(v))\ge \alpha$).
Hence, each middle palindromic substring is induced by some spine-tree. 

\smallskip
The spine-trees are not disjoint, but, nevertheless, their total size is small.
\begin{lemma}
The sizes $n_1,\ldots, n_k$ of the created spine-trees satisfy $\sum_i n_i \le 2n$.
\label{lem:sizes}
\end{lemma}
\begin{proof}
We claim that 
at least $n_i-\alpha$ nodes of the $i$th spine-tree are
disjoint from all the other spine-trees. 
Let $c_i$ be the node on the spine of the $i$th spine-tree such that $\dist(c_i,r)=\sqrt{n}$
and similarly let $s_i$ satisfy $\dist(s_i,r)=\alpha$. Recall that $\per(s_i,r)\leq \frac12\sqrt{n}$.
Thus, \cref{lem:periodicperiod} yields $\per(s_i,r)=\per(c_i,r)$. Since the tree is deterministic,
$c_i$ uniquely determines $s_i$ and hence the whole spine-tree. Thus, the nodes $c_i$ are all distinct and so are their
predecessors on the spines and all attached subtrees. 

\medskip A similar argument shows that all nodes $d_i$ on the spine
of the $i$th spine-tree such that $\dist(r,d_i)=\sqrt{n}$ are also all distinct.
Therefore, we proved $\sum_i n_i-\alpha \leq n$.

\medskip
Each spine-tree has at least $2\alpha$ vertices on the spine, so this yields $n_i\ge 2\alpha$, and thus we obtain
\[\sum_i n_i \leq 2\sum_i (n_i-\alpha)\leq 2n.\qedhere\]
\end{proof}

\begin{lemma}
\label{lem:decompose}
Every D-tree of size $n$ has $\Oh(n^{1.5})$ distinct palindromic substrings.
\end{lemma}
\begin{proof}
By \cref{lem:spine-trees}, the number of palindromes induced by the $i$th spine-tree is at most $n_i^{1.5}$.
Accounting the $\Oh(n^{1.5})$ palindromes which do not occur as middle palindromes, we have 
\[\Oh(n^{1.5})+\sum_i n_i^{1.5}\le \Oh(n^{1.5})+\sum_i n_i \sqrt{n} = \Oh(n^{1.5})\] palindromes in total.
\end{proof}

Due to Lemma~\ref{centroids} and Fact~\ref{Psi}  we obtain.
\begin{theorem}\label{thm:main}
A tree with $n$ edges contains $\Oh(n^{1.5})$ distinct palindromic substrings.
\end{theorem}

\section{Algorithm Reporting All Distinct Palindromes}
\label{sec:alg}

In this section, we consider following problem for palindromes in trees:

\begin{problem}[\ReportAllPalindromes]
    Given tree a $T$ with $n$ edges, each labelled by single character from the alphabet $\Sigma$
	report all distinct palindromes in $T$.
\end{problem}

There are various ways for reporting palindromes,
the natural choice is to represent each palindrome
as a pair for nodes $(u,v)$ such that $val(u,v)$ is a palindrome.
Unfortunately for efficiency reasons we would like to use
slightly different format, each palindrome will
be reported as triple $(\ell, u, v)$, such
that $\ell$ is length of a palindrome, $val(u,v)$ is a first half of a palindrome
($\left\lceil \frac{\ell}{2}\right\rceil=|val(u,v)|$).

\noindent
To simplify the description of the algorithm and introduce
restricted version of the problem:

\begin{problem}[\ReportAllEvenPalindromes]
    Given tree $T$ with $n$ edges, each labelled by single character from the alphabet $\Sigma$
    report all distinct even palindromes in $T$.
\end{problem}

\noindent
The following lemma states that in fact
the problem $\ReportAllEvenPalindromes$
is equivalent to the problem $\ReportAllPalindromes$.

\begin{lemma}\label{all-to-even}
    Given an algorithm for problem $\ReportAllEvenPalindromes$ running in time $f(n)$,
    it is possible to solve problem $\ReportAllPalindromes$ in $f(O(n))$ time.
\end{lemma}

\begin{proof}
    Given an instance $T$ of the problem $\ReportAllPalindromes$, we can generate tree $T'$
    by replacing each edge $(u,v)$ with label $c\in \Sigma$ from $T$ by a
    path of length 4 with corresponding labels $\$,c,c,\$$ (where $\$$ is a character not in $\Sigma$).
    Each palindrome in $T$ has a corresponding palindrome in $T'$.
    Also each palindrome of length $4k$ that starts (and ends) with character $\$$ in $T'$ can be attributed
    to corresponding palindrome in $T$.
    In consequence we can solve the problem $\ReportAllEvenPalindromes$ for $T'$ and report only those
    palindromes in $T$ that corresponds to the palindromes of length $4k$ that start with character $\$$.
\end{proof}

\subsection{Algorithmic tools}

Before we describe the algorithm we introduce the algorithmic toolbox used in our algorithm.
It is similar to the one described in the Section 3 from~\cite{DBLP:journals/tcs/KociumakaPRRW14},
but it is tailored to the palindromic case.

\begin{lemma}
    \label{lem:tools_linear}
    Given a family of D-trees $D_1,\ldots,D_k$ with total $n$ nodes.
    It can be preprocessed in $\Oh(n)$ time, such that following operations can be done in $\Oh(1)$ time:
    \begin{itemize}
        \item $\distFunc(u, v)$ -- distance between nodes $u$ and $v$,
        \item $\upFunc(u, h)$ -- node $v$ on a path from $u$ towards the root, at distance $h$ from $u$,
        \item $\centerFunc(u, v)$ -- node at the center of path from $u$ and $v$,
        \item $\isAncestor(u, v)$ -- is $v$ an ancestor of $u$,
        \item $\perLen(u)$ -- length of the period of word on path from root to $u$
    \end{itemize}
\end{lemma}

\begin{proof}
    Queries $\distFunc(u, v)$ can be implemented by precomputing depth of each node in a tree
    and using Lowest Common Ancestor Queries (LCA)~\cite{DBLP:journals/siamcomp/HarelT84}.
    Query $\upFunc(u, h)$ is in fact Level Ancestor Query (LA)~\cite{DBLP:journals/tcs/BenderF04}.
    Operation $\centerFunc(u, v)$ can be realized by one $\distFunc$ and $up$ query.
    Length of periods $\perLen(u)$ can be calculates from the border array $P$ that can be computed
    in $\Oh(n)$ time (\cite{DBLP:journals/tcs/KociumakaPRRW14}).
\end{proof}

\begin{lemma}
    \label{lem:tools_nlogn}
    Given a family of D-trees $D_1,\ldots,D_k$ with total $n$ nodes.
    It can be preprocessed in $\Oh(n\log n)$ time, such that following operations can be done in $\Oh(1)$ time:

    \begin{itemize}
    \item $\labelFunc(u, v)$ -- pair of integers representing word $val(u, v)$
        (this operation is defined only for $u$ being ancestor of $v$ or $v$ being ancestor of $u$),
    \item $\isEqual(u_1, v_1, u_2, v_2)$ -- is $val(u_1, v_1)=val(u_2, v_2)$,
    \item $\isPalindrome(u, v)$ -- is word $val(u, v)$ a palindrome?,
    \item $\existsFunc(D_i, u, v)$ -- for $D_i=(L_i, R_i, r_i)$ it verifies if there
        exist a node $w\in R_i$ such that $val(w, r_i)=val(u, v)$
        (this operation is defined only for $u,v\in L_i$ and $u$ being ancestor of $v$ or $v$ being ancestor of $u$),
    \item $\childFunc(u, c)$ -- returns child node of $u$ with label $c$ (or null values if it does not exist)
    \end{itemize}
\end{lemma}
\begin{proof}
    Operations $\labelFunc$, $\isEqual$ and $\isPalindrome$ can be implemented using
    Dictionary of Basic Factors (DBF)~\cite{Jewels}, which clearly requires $\Oh(n\log n)$ preprocessing time.
    The only extension is that we need is that for each basic factor we also store code of its reversed version.
    For operation $\existsFunc$ we store DBF codes of all possible values of $val(r, w)$ in a static dictionary
    with constant lookup time.
\end{proof}

\subsection{Algorithm for spine trees}

\begin{lemma}
    \label{lem:spine-decomposition}
    For a double tree $D=(T_l, T_r, r)$ with $n$ nodes, the spine decomposition
    can be calculated in $\Oh(n)$ time, assuming that the $D$ has been preprocessed with Lemma \ref{lem:tools_linear}.
\end{lemma}

\begin{proof}
The spine decomposition of double tree $D=(T_l, T_r, r)$
We start with calculating set $C$ with nodes from $T_l$ at distance $2\sqrt{n}$ from root
with value $\perLen(u) < \frac{1}{2}\sqrt{n}$.
Since $D$ has been preprocessed with Lemma \ref{lem:tools_linear} the set $C$ can be calculated in $\Oh(n)$ time.
We can observe that for any nodes $u_1,u_2\in C$ ($u_1\not=u_2$)
due to high periodicity of $val(u_1, r)$ and $val(u_2, r)$, the paths $(u_1, r)$ and $(u_2, r)$ have
at least $\sqrt{n}$ distinct nodes, so $|C| < \sqrt{n}$.

Next, for each candidate node $u\in C$, we need to verify if it is a part of a spine.
Let $p$ a string period of $val(u, r)$, we locate lowest descendant $S_l$ of $u$ such that string period
of $val(S_l, r)$ is $p$. Such node can be located by traversing subtree of $u$ with $child(x, c)$ queries.
Similarly we traverse $T_r$ starting from root $r$ to locate lowest node in $S_r$ such that string period
of $val(r, S_r)$ is $p^R$.
If $\distFunc(r, S_r) \ge 2\sqrt{n}$ we add to the result spine tree with spine $(S_l, S_r)$
and all subtrees with distance $\ge 2\sqrt{n}$ attached.

In this procedure only edges in $T_r$ with distance $<\sqrt{n}$ can be visited multiple times, but since
$|C|<\sqrt{n}$ the total processing time of such edges is still $\Oh(n)$.
\end{proof}

For efficient processing spine trees, we need one additional lemma:

\begin{lemma}{\rm [ {\bf FFT Application}] }
    \label{lem:fft}
    Given two set of integers $A, B\subset [0,\ldots,n]$
    the set $A\ominus B=\{ a - b : a\in A, b\in B\}$
    can be computed in $\Oh(n\log n)$ time.
\end{lemma}

\begin{proof}
    We define two polynomials
    $$f_A(x) = \sum_{i\in A} x^{n + i},\ \ \
    f_B(x) = \sum_{i\in B} x^{n - i}$$
    Using FFT we can multiply two polynomials
    with integer coefficients in time $\Oh(n\log n)$ (\cite{DBLP:books/daglib/0023376}).
    And clearly polynomial
    $f(x)=f_A(x) \cdot f_B(x)$ has non-zero $i$-th coefficient
    iff $i-2n\in A\ominus B$.
\end{proof}

\begin{lemma}
    For a spine tree $S=(T_l, T_r, r)$ with $n$ nodes
    it is possible to calculate in time $\Oh(n^{1.5}\log n)$,
    the set of even palindromes $P$ in $S$ such that:
    \begin{itemize}
        \item $|P|=\Oh(n^{1.5})$,
        \item $P$ contains all even palindromes in $S$
            with left endpoint and middle point in $T_l$ and right endpoint in $T_r$.
    \end{itemize}
    \label{lem:alg-pal-in-spine-trees}
\end{lemma}

\begin{proof}
    First, let us remind that the complexity of the algorithm is very close to the actual limit on
    the number of distinct palindromes in spine tree, since there could be up to $\Oh(n^{1.5)}$ palindromes in $S$
    (see~Lemma \ref{lem:spine-trees}).
    Our approach is very similar to the one used in the proof of~Lemma \ref{lem:spine-trees}.
    We identify the labels $L(u)$ for each $u\in S$ with a distance at least $2\sqrt{n}$ from the root.
    Since the tree is already preprocessed with~Lemma \ref{lem:tools_linear} and~Lemma \ref{lem:tools_nlogn}
    such labels can be retrieved and represented in constant time and space.

    Next all labels $L(u)$ are sorted in $\Oh(n\log n)$ time and we group all nodes with the same label
    into groups $V_{L_1}$, $V_{L_2}, \ldots, V_{L_k}$.

    For each group $V_L$ with at most $\sqrt{n}$ nodes can be inspected in $\Oh(|V_L|^2)$ time,
    for each $x \in V_L\cap T_l$ and $y\in V_L\cap T_r$ we check in $\Oh(1)$ the condition $\isPalindrome(x, y)$
    and report the palindrome if the condition is true.

    For each group $V_L$ with more than $\sqrt{n}$ nodes we will use discrete convolutions
    to speed up the calculations.
    First we need to verify if the spine part of the palindromes from $V_L$ is in fact palindromic.
    This can be checked by locating any pair of nodes $x \in V_L\cap T_l, y\in V_L\cap T_r$ such that
    $\distFunc(x, y)$ is even. If the condition $\isPalindrome(x, y)$ is true, then for any pair of nodes
    $x \in V_L\cap T_l, y\in V_L$ with even distance we obtain palindrome, all we need to do is
    to identify all possible (even) values of $\distFunc(s(x), s(y))$.

    Let $s_\ell$ be the left endpoint of the spine of $S$
    and $$X_L=\{\distFunc(s_\ell, s(x))\ \mbox{for}\ x\in V_L\cap T_l\}$$
    $$Y_L=\{\distFunc(s_\ell, r) + \distFunc(s(y), r)\ \mbox{for}\ y\in V_L\cap T_r\}.$$
    The set of all possible differences $\Delta_L$ can be obtained by computing $Y_L\ominus X_L$ and
    taking only even values.
    This step takes $\Oh(n\log n)$ due to Lemma \ref{lem:fft}.

    Unfortunately using this step we don't have a witnesses for values $\delta \in \Delta$.
    Nevertheless we are able to reconstruct the palindromic substrings itself.
    Let $x_0$ is the node from $V_L\cap T_l$ that is the farthest from the root $r$.
    For each even value $\delta\in \Delta$ with $\frac{\delta}{2} \le \distFunc(s(x_0), r)$
    we report palindrome with value $w w^R$ where $w=val(x_0, up(s(x), \frac{\delta}{2}))$.
    Please not that this palindrome might not occur in the node $x_0$,
    also we might over-report here and report also the palindromes that have a middle point in $T_r$.
\end{proof}

\subsection{Algorithm for general D-trees}

In our algorithm we will use similar approach to the one used in the proof of~Lemma \ref{lem:decompose}.
There are a few technical issues that we need to overcome.
First we need to strengthen a notion of D-trees,
we need to make sure that all paths in D-tree
correspond to simple paths in the original tree.
Unfortunately due to repeated edges adjacent to the root this rule can be violated
(see~Figure~\ref{fig:alg_double_tree_issue}).
Second problem is efficient calculation of palindromes in spine trees.

\begin{figure}[ht]
    \vspace{0.1cm}
	\begin{center}
		\begin{tikzpicture}[
    scale=0.85,
    dot/.style={draw,circle,minimum size=2mm,inner sep=0pt,outer sep=0pt,fill=black},
    leftArr/.style={latex-,thick},
    rightArr/.style={-latex,thick}
  ]
  
  \begin{scope}
  \node (r) [circle, draw, fill=gray!50!white] at (0,0) {$r$};
  \node (a) [dot] at (-1,1) {};
  \node (a1) [dot] at (-2,1.5) {};
  \node (a2) [dot] at (-2,0.5) {};
  \node (a3) [dot] at (-3,1.5) {};
  \node (a4) [dot] at (-4,1.5) {};
  \node (b) [dot] at (-1,-1) {};
  
  \draw (r)--(a) node[midway,above] {$a$};
  \draw (r)--(b) node[midway,above] {$b$};
  \draw (a)--(a1) node[midway,above] {$a$};
  \draw (a)--(a2) node[midway,below] {$b$};
  \draw (a1)--(a3) node[midway,above] {$a$};
  \draw (a3)--(a4) node[midway,above] {$b$};
  
  \draw[dotted] (1,2)--(1,-2);
  \end{scope}
  
  \begin{scope}[xshift=6cm]
  \node (lr) [circle, draw, fill=gray!50!white] at (0,0) {$r$};
  \node (la) [dot] at (-1,1) {};
  \node (la1) [dot] at (-2,1.5) {};
  \node (la2) [dot] at (-2,0.5) {};
  \node (la3) [dot] at (-3,1.5) {};
  \node (la4) [dot] at (-4,1.5) {};
  \node (lb) [dot] at (-1,-1) {};
  \node (ra) [dot] at (1,1) {};
  \node (ra1) [dot] at (2,1.5) {};
  \node (ra2) [dot] at (2,0.5) {};
  \node (ra3) [dot] at (3,1.5) {};
  \node (ra4) [dot] at (4,1.5) {};
  \node (rb) [dot] at (1,-1) {};
  
  \draw[leftArr] (lr)--(la) node[midway,above] {$a$};
  \draw[leftArr] (lr)--(lb) node[midway,above] {$b$};
  \draw[leftArr] (la)--(la1) node[midway,above] {$a$};
  \draw[leftArr] (la)--(la2) node[midway,below] {$b$};
  \draw[leftArr] (la1)--(la3) node[midway,above] {$a$};
  \draw[leftArr] (la3)--(la4) node[midway,above] {$b$};
  \draw[rightArr] (lr)--(ra) node[midway,above] {$a$};
  \draw[rightArr] (lr)--(rb) node[midway,above] {$b$};
  \draw[rightArr] (ra)--(ra1) node[midway,above] {$a$};
  \draw[rightArr] (ra)--(ra2) node[midway,below] {$b$};
  \draw[rightArr] (ra1)--(ra3) node[midway,above] {$a$};
  \draw[rightArr] (ra3)--(ra4) node[midway,above] {$b$};
  \end{scope}

  \end{tikzpicture}
	\end{center}
    \vspace{0.1cm}
	\caption{Undirected tree $T$ and its D-tree $D=(T_l,r,T_r)$. Note that $D$ contains path
	with even palindrome $baaaab$ that is not present in the original tree $T$.}
	\label{fig:alg_double_tree_issue}
\end{figure}
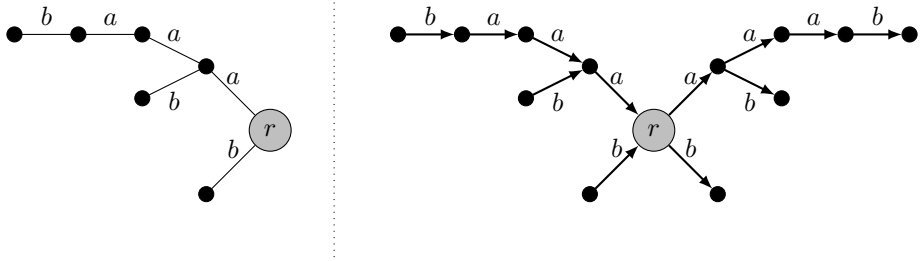

\noindent
We resolve those issues with 2D-trees using following lemma:

\begin{lemma}
	\label{lem:alg-decomposition-of-tree-to-double-trees}
    Given undirected tree $T$ with $n$ nodes, we can calculate decomposition
    of $T$ into family of D-trees $\mathcal{D}$ such that:
	\begin{itemize}
		\item each simple path in $D\in\mathcal{D}$ corresponds to a simple path in $T$,
		\item for each even path $p\in T$, there exists $D\in\mathcal{D}$ such that
			there exists even path $p'\in D$ such that $val(p)=val(p')$ and middle point of $p'$ is in the left subtree of $D$,
		\item total number of edges in all double trees from $\mathcal{D}$ is $\Oh(n\log n)$.
	\end{itemize}
	The decomposition $\mathcal{D}$ can be calculated in time $\Oh(n\log n)$.
\end{lemma}

\begin{proof}
    The decomposition $\mathcal{D}$ can be created in recursive manner.

    \noindent
    For a tree $T$, we use following procedure:
	\begin{itemize}
		\item identify centroid node $r$,
        \item divide subtrees adjacent to $r$ into two trees $T_1$, $T_2$,
            such that $\max(|T_1|,|T_2|) \le \frac{3}{4}|T|$,
		\item create deterministic version of trees $T_1$ and $T_2$ $T'_1$ and $T'_2$,
		\item to handle paths that have one endpoint in $T_1$ and other in $T_2$ we add
			to $\mathcal{D}$ D-trees $(T'_1, T'_2, r)$ and $(T'_2, T'_1, r)$,
		\item to handle paths that are contained in $T_1$ or $T_2$
			recursively process decomposition of $T_1$ and $T_2$.
    \end{itemize}

    \noindent
    The total size of the created decomposition is $\Oh(n\log n)$.
\end{proof}

Now we are ready to outline the algorithm.
For given tree $T$ we decompose it into family or D-trees $\mathcal{D}$.

\noindent
The family of trees $\mathcal{D}$ is preprocessed using Lemma \ref{lem:tools_linear} and Lemma \ref{lem:tools_nlogn}.
We need to process all trees from $\mathcal{D}$ altogether to obtain consistent DFS identifiers
between different trees.

\noindent
Each double tree $D\in\mathcal{D}$ is processed separately,
we find the middle palindromes using spine decomposition
and all other palindromes using exhaustive search.

\noindent
Finally we remove possible duplicates in reported palindromes using sorting.

\vspace*{0.7cm}
\begin{algorithm}[H]
    \DontPrintSemicolon
    \caption{{\sc FindPalindromesInDoubleTree}$(D=(T_l, T_r, r))$}
    preprocess the tree $D$ for queries from lemmas~\ref{lem:tools_linear} and \ref{lem:tools_nlogn} \;
    $P=\emptyset$ \;
    // handle middle palindromes \;
    decompose $D$ into spine trees $S_1,\ldots,S_k$ \;
    \ForEach{$S_i\in \mbox{\sc SpineDecomposition}(D)$}{
        add to $P$ palindromes reported by Lemma \ref{lem:alg-pal-in-spine-trees} in $S_i$ \;
    }
    // handle first (shorter) and last (longer) palindromes \;
    \ForEach{$u\in T_l$}{
        let $L_u$ is a data structure that allows access to the list of ancestors $u'$ of $u$ such that $val(u', r)$ is a palindrome \;
        elements of $L_u$ are sorted by their depth in a tree \;
        \ForEach{$u'\in (\mbox{first}\ 4\sqrt{n}\ \mbox{nodes from}\ L_u) \cup (\mbox{last}\ 2\sqrt{n}\ \mbox{nodes from}\ L_u)$}{
            if there exists node $v\in T_r$ such that $val(v, r)=val(u, u')$ and $val(u, v)$ is a palindrome
            add it to $P$
        }
    }
    \KwRet{$P$}
    \label{alg:FindPalindromesInDoubleTree}
\end{algorithm}
\vspace{0.2cm}

\begin{lemma}
    \label{lem:alg-FindPalindromesInDoubleTree}
    For a D-tree $D$ with $n$ nodes the problem
    $\ReportAllEvenPalindromes$ can be solved in
    in $\Oh(n^{1.5}\log n)$ time.
\end{lemma}

\begin{proof}
    The pseudocode of the solution is given in Algorithm~\ref{alg:FindPalindromesInDoubleTree}.
    The correctness of the algorithm is proved by the Lemma \ref{lem:decompose}.
    We need to prove that the algorithm can be implemented in $\Oh(n^{1.5}\log n)$ time.
    Calculating the spine decomposition requires $\Oh(n)$ time due to Lemma \ref{lem:spine-decomposition},
    and each spine $S$ can be processed in $\Oh(|S|^{1.5}\log |S|)$ time due to Lemma \ref{lem:alg-pal-in-spine-trees}.
    Since total size of the spine trees is $\Oh(n)$, this part takes $\Oh(n^{1.5}\log n)$.

    For handling the first and last palindromes, we need a data structure to operate on lists $L_u$.
    We identify all nodes $X=\{x \in D : val(x, r)\ \mbox{is a palindrome}\}$ using
    Lemma \ref{lem:tools_nlogn}.
    Then we create a subtree $D_X$ which contains only nodes from $X$ and preprocess it for
    Level Ancestor queries (\cite{DBLP:journals/tcs/BenderF04}).
    Additionally for each $u\in D$ we store its nearest ancestor in $D_X$ and its depth in $D_x$ (equal to $L_u$).
    With this approach any element of $L_u$ can be retrieved in $\Oh(1)$ time using LA queries on $D_x$.

    \noindent
    For each element $u'\in L_u$ we test the existence of node $w$ in $\Oh(1)$ time using function $\existsFunc$.
\end{proof}

\begin{theorem}\mbox{ \ }\\
    \label{thm:alg-report-all-even-palindromes}
    For a a tree $T$ with $n$ nodes the problem $\ReportAllPalindromes$ can be solved in
    $\Oh(n^{1.5}\log n)$ time.
\end{theorem}

\begin{proof}
Due to Lemma~\ref{all-to-even} the problem reduces to counting even palindromes.
Hence we later consider only even palindromes.

    First we decompose the tree $T$ into set of 
    D-trees $\mathcal{D} = D_1,\ldots,D_k$ using Lemma \ref{lem:alg-decomposition-of-tree-to-double-trees}.
    All trees are preprocessed using tools from Lemma \ref{lem:tools_linear} and Lemma \ref{lem:tools_nlogn}.

    Next we calculate palindromes in all D-trees $D_i\in\mathcal{D}$
    using Lemma \ref{alg:FindPalindromesInDoubleTree}.
    Due to construction of $\mathcal{D}$ this
    requires time $$T(n) = T(\alpha n) + T((1-\alpha) n) + \Oh(n^{1.5}\log n)$$
    (for $\frac{1}{4} \le \alpha \le \frac{3}{4}$),
    which is $T(n)=\Oh(n^{1.5}\log n)$ and the total size
    of returned palindromes $\mathcal{P}$ is $$P(n) = P(\alpha n) + P((1-\alpha)n) + \Oh(n^{1.5})$$
    (for $\frac{1}{4} \le \alpha \le \frac{3}{4}$)
    which is $P(n)=\Oh(n^{1.5})$.

    Finally we remove from $\mathcal{P}$ duplicates.
    Since all palindromes are identified by length, and pair
    of integers generated by $\labelFunc$, we can sort $\mathcal{P}$ in $\Oh(|P|)$ time
    which is $\Oh(n^{1.5})$.
\end{proof}

\section{Algorithm for finding longest palindrome in tree}
\label{sec:alg_t}

In this section, we consider the following problems for palindromes in trees:

\begin{problem}[\PalindromeTesting]
    Given tree a $T$ with $n$ edges, each labelled by single character from the alphabet $\Sigma$
    and integer $k>0$,
	decide whatever $T$ contains palindrome of length exactly $k$.
\end{problem}

\begin{problem}[\FindLongestPalindrome]
    Given tree a $T$ with $n$ edges, each labelled by single character from the alphabet $\Sigma$
	find the length of the longest palindrome in $T$.
\end{problem}

\begin{theorem}\mbox{ \ }\\
    {\bf (a)} Problem $\PalindromeTesting$ can be solved in $O(n\log^2 n)$ time.\\
    {\bf (b)}
    Problem $\FindLongestPalindrome$ can be solved in $O(n\log^3 n)$ time.
\end{theorem}
\begin{proof}
The point {\bf (a)} can be solved using the algorithm below.
The point {\bf (b)} can be solved using the testing function
together with
binary search.
\end{proof}

\begin{algorithm}[H]
    \caption{Test if there exists palindrome of length $k$ in DD-tree $D$}
    \ForEach{$u\in nodes(T)$}{
        \eIf{$depth(u)\ge k$}{
            let $v$ is a node $up(u, k)$ \;
            \If{$val(u, v)$ is a palindrome}{
                \KwRet{true}
            }
        }{
            let $v'$ is a node $up(u, depth(u)-k)$ \;
            \If{$\isPalindrome(v',r)$ {\bf and} $\existsFunc(D, u, v')$}{
                \KwRet{true}
            }
        }
    }
    \KwRet{false}
\end{algorithm}

\section{Open problems}

We conclude with following open question:
\begin{itemize}
	\item is there an output sensitive version of the all palindromes reporting -- can we report palindromes
		more efficiently for cases where we know that tree contains $o(n^{1.5})$ palindromes.
\end{itemize}

\section*{Acknowledgments}
  Paweł Gawrychowski's work was done while he held a post-doctoral position at Warsaw
    Center of Mathematics and Computer Science.
  Tomasz Kociumaka was supported by Polish budget funds for science in 2013--2017 as a research project under the `Diamond Grant' program.
  Wojciech Rytter was supported by the Polish National Science Center, grant no NCN2014/13/B/ST6/{\linebreak}00770.
  Tomasz Waleń was supported by the Polish Ministry of Science and Higher Education under the `Iuventus Plus' program in 2015--2016 grant no 0392/IP3/2015/73.

\bibliographystyle{plainurl}
\bibliography{pal.bib}

\begin{thebibliography}{10}

\bibitem{DBLP:journals/dm/AnisiuAK10}
Mira{-}Cristiana Anisiu, Valeriu Anisiu, and Zolt{\'{a}}n K{\'{a}}sa.
\newblock Total palindrome complexity of finite words.
\newblock {\em Discrete Mathematics}, 310(1):109--114, 2010.
\newblock \href {http://dx.doi.org/10.1016/j.disc.2009.08.002}
  {\path{doi:10.1016/j.disc.2009.08.002}}.

\bibitem{DBLP:journals/tcs/BenderF04}
Michael~A. Bender and Martin Farach{-}Colton.
\newblock The level ancestor problem simplified.
\newblock {\em Theor. Comput. Sci.}, 321(1):5--12, 2004.
\newblock URL: \url{https://doi.org/10.1016/j.tcs.2003.05.002}, \href
  {http://dx.doi.org/10.1016/j.tcs.2003.05.002}
  {\path{doi:10.1016/j.tcs.2003.05.002}}.

\bibitem{DBLP:journals/ijfcs/BrlekHNR04}
Srečko Brlek, Sylvie Hamel, Maurice Nivat, and Christophe Reutenauer.
\newblock On the palindromic complexity of infinite words.
\newblock {\em International Journal of Foundations of Computer Science},
  15(2):293--306, 2004.
\newblock \href {http://dx.doi.org/10.1142/s012905410400242x}
  {\path{doi:10.1142/s012905410400242x}}.

\bibitem{Brlek}
Srečko Brlek, Nadia Lafrenière, and Xavier Provençal.
\newblock Palindromic complexity of trees.
\newblock In Igor Potapov, editor, {\em Developments in Language Theory, DLT
  2015}, volume 9168 of {\em LNCS}, pages 155--166. Springer, 2015.
\newblock \href {http://dx.doi.org/10.1007/978-3-319-21500-6_12}
  {\path{doi:10.1007/978-3-319-21500-6_12}}.

\bibitem{DBLP:books/daglib/0023376}
Thomas~H. Cormen, Charles~E. Leiserson, Ronald~L. Rivest, and Clifford Stein.
\newblock {\em Introduction to Algorithms, 3rd Edition}.
\newblock {MIT} Press, 2009.
\newblock URL: \url{http://mitpress.mit.edu/books/introduction-algorithms}.

\bibitem{DBLP:conf/cpm/CrochemoreIKKRRTW12}
Maxime Crochemore, Costas~S. Iliopoulos, Tomasz Kociumaka, Marcin Kubica, Jakub
  Radoszewski, Wojciech Rytter, Wojciech Tyczyński, and Tomasz Waleń.
\newblock The maximum number of squares in a tree.
\newblock In Juha K{\"a}rkk{\"a}inen and Jens Stoye, editors, {\em
  Combinatorial Pattern Matching}, volume 7354 of {\em LNCS}, pages 27--40.
  Springer, 2012.
\newblock \href {http://dx.doi.org/10.1007/978-3-642-31265-6_3}
  {\path{doi:10.1007/978-3-642-31265-6_3}}.

\bibitem{Jewels}
Maxime Crochemore and Wojciech Rytter.
\newblock {\em Jewels of Stringology}.
\newblock World Scientific, 2003.

\bibitem{DBLP:journals/tcs/DroubayJP01}
Xavier Droubay, Jacques Justin, and Giuseppe Pirillo.
\newblock Episturmian words and some constructions of de {L}uca and {R}auzy.
\newblock {\em Theoretical Computer Science}, 255(1-2):539--553, 2001.
\newblock \href {http://dx.doi.org/10.1016/s0304-3975(99)00320-5}
  {\path{doi:10.1016/s0304-3975(99)00320-5}}.

\bibitem{fine1965uniqueness}
Nathan~J. Fine and Herbert~S. Wilf.
\newblock Uniqueness theorems for periodic functions.
\newblock {\em Proceedings of the American Mathematical Society},
  16(1):109--114, 1965.
\newblock \href {http://dx.doi.org/10.2307/2034009}
  {\path{doi:10.2307/2034009}}.

\bibitem{DBLP:conf/stringology/FunakoshiNIBT19}
Mitsuru Funakoshi, Yuto Nakashima, Shunsuke Inenaga, Hideo Bannai, and Masayuki
  Takeda.
\newblock Computing maximal palindromes and distinct palindromes in a trie.
\newblock In Jan Holub and Jan Zd{\'{a}}rek, editors, {\em Prague Stringology
  Conference 2019, Prague, Czech Republic, August 26-28, 2019}, pages 3--15.
  Czech Technical University in Prague, Faculty of Information Technology,
  Department of Theoretical Computer Science, 2019.
\newblock URL: \url{http://www.stringology.org/event/2019/p02.html}.

\bibitem{DBLP:conf/spire/GawrychowskiKRW15}
Pawel Gawrychowski, Tomasz Kociumaka, Wojciech Rytter, and Tomasz Walen.
\newblock Tight bound for the number of distinct palindromes in a tree.
\newblock In Costas~S. Iliopoulos, Simon~J. Puglisi, and Emine Yilmaz, editors,
  {\em String Processing and Information Retrieval - 22nd International
  Symposium, {SPIRE} 2015, London, UK, September 1-4, 2015, Proceedings},
  volume 9309 of {\em Lecture Notes in Computer Science}, pages 270--276.
  Springer, 2015.
\newblock URL: \url{https://doi.org/10.1007/978-3-319-23826-5\_26}, \href
  {http://dx.doi.org/10.1007/978-3-319-23826-5\_26}
  {\path{doi:10.1007/978-3-319-23826-5\_26}}.

\bibitem{DBLP:journals/ejc/GlenJWZ09}
Amy Glen, Jacques Justin, Steve Widmer, and Luca~Q. Zamboni.
\newblock Palindromic richness.
\newblock {\em European Journal of Combinatorics}, 30(2):510--531, 2009.
\newblock \href {http://dx.doi.org/10.1016/j.ejc.2008.04.006}
  {\path{doi:10.1016/j.ejc.2008.04.006}}.

\bibitem{GlenSS19}
Amy Glen, Jamie Simpson, and William~F. Smyth.
\newblock Palindromes in starlike trees.
\newblock {\em Australasian Journal of Combinatorics}, 73(1):242--246, 2019.
\newblock URL: \url{https://ajc.maths.uq.edu.au/pdf/73/ajc_v73_p242.pdf}.

\bibitem{DBLP:journals/ipl/GroultPR10}
Richard Groult, {\'{E}}lise Prieur, and Gw{\'{e}}na{\"{e}}l Richomme.
\newblock Counting distinct palindromes in a word in linear time.
\newblock {\em Information Processing Letters}, 110(20):908--912, 2010.
\newblock \href {http://dx.doi.org/10.1016/j.ipl.2010.07.018}
  {\path{doi:10.1016/j.ipl.2010.07.018}}.

\bibitem{DBLP:journals/siamcomp/HarelT84}
Dov Harel and Robert~Endre Tarjan.
\newblock Fast algorithms for finding nearest common ancestors.
\newblock {\em {SIAM} J. Comput.}, 13(2):338--355, 1984.
\newblock URL: \url{https://doi.org/10.1137/0213024}, \href
  {http://dx.doi.org/10.1137/0213024} {\path{doi:10.1137/0213024}}.

\bibitem{DBLP:journals/tcs/KociumakaPRRW14}
Tomasz Kociumaka, Jakub Pachocki, Jakub Radoszewski, Wojciech Rytter, and
  Tomasz Walen.
\newblock Efficient counting of square substrings in a tree.
\newblock {\em Theor. Comput. Sci.}, 544:60--73, 2014.
\newblock URL: \url{https://doi.org/10.1016/j.tcs.2014.04.015}.

\bibitem{DBLP:journals/algorithmica/KociumakaRRW17}
Tomasz Kociumaka, Jakub Radoszewski, Wojciech Rytter, and Tomasz Waleń.
\newblock String powers in trees.
\newblock {\em Algorithmica}, 79(3):814--834, 2017.
\newblock \href {http://dx.doi.org/10.1007/s00453-016-0271-3}
  {\path{doi:10.1007/s00453-016-0271-3}}.

\bibitem{DBLP:journals/tcs/Simpson14}
Jamie Simpson.
\newblock Palindromes in circular words.
\newblock {\em Theoretical Computer Science}, 550:66--78, 2014.
\newblock \href {http://dx.doi.org/10.1016/j.tcs.2014.07.012}
  {\path{doi:10.1016/j.tcs.2014.07.012}}.

\end{thebibliography}

\end{document}